\documentclass[nopreprintline,preprint,12pt]{elsarticle}
\usepackage[T1]{fontenc}
\usepackage{amsmath,amssymb,amsthm,mathtools}
\usepackage{array}
\usepackage{graphicx}
\usepackage[colorlinks=true,citecolor=blue,linkcolor=blue,urlcolor=blue]{hyperref}
\hypersetup{
  pdftitle={Hankel Submatrices with Bounded Word Length: Dimension Tradeoffs in Probabilistic and Quantum Automata},
  pdfauthor={Yixin Zhao},
  pdfsubject={Dimension tradeoffs for Hankel submatrices generated by probabilistic and quantum automata},
  pdfkeywords={bounded word length, weighted automata, Hankel matrices, probabilistic automata, quantum finite automata}
}

\newcommand{\Tr}{\operatorname{Tr}}
\newcommand{\supp}{\operatorname{supp}}
\newcommand{\diag}{\operatorname{diag}}
\newcommand{\Ent}{\operatorname{Ent}}
\newcommand{\cW}{\mathcal W}
\newcommand{\cA}{\mathcal A}
\newcommand{\cC}{\mathcal C}
\newcommand{\cX}{\mathcal X}
\newcommand{\one}{\mathbf 1}

\newtheorem{theorem}{Theorem}[section]
\newtheorem{proposition}[theorem]{Proposition}
\newtheorem{lemma}[theorem]{Lemma}
\newtheorem{corollary}[theorem]{Corollary}
\theoremstyle{definition}
\newtheorem{definition}[theorem]{Definition}
\theoremstyle{remark}
\newtheorem{remark}[theorem]{Remark}

\begin{document}
%
%

\begin{frontmatter}

\title{Hankel Submatrices with Bounded Word Length: Dimension Tradeoffs in Probabilistic and Quantum Automata}

\author{Yixin Zhao\corref{cor1}}
\cortext[cor1]{Corresponding author.}
\address{Beijing Key Laboratory of Fault-Tolerant Quantum Computing,
Beijing Academy of Quantum Information Sciences, Beijing 100193, China;
Beijing National Laboratory for Condensed Matter Physics, Institute of
Physics, Chinese Academy of Sciences, Beijing 100190, China; and University
of Chinese Academy of Sciences, Beijing 100190, China}
\ead{zhaoyixin22@mails.ucas.ac.cn}

\begin{abstract}
We characterize the largest admissible side length of fixed-margin square
Hankel submatrices generated by probabilistic realizations on a \(d\)-point
carrier and quantum linear realizations on a \(d\)-dimensional Hilbert space,
when prefix and suffix word lengths are bounded by \(T\). Using a four-symbol
transition alphabet shared by all selected rows and columns, we show that, for
fixed margins \(\lambda_2>0\) and \(\lambda_1+\lambda_2<1\), the corresponding
side lengths satisfy
\(\ln N_{\mathsf Q,4}=\Theta_{\lambda_1,\lambda_2}(\min\{d^2,T\})\)
for completely positive trace-preserving (CPTP) quantum realizations and
\(\ln N_{\mathsf S,4}=\Theta_{\lambda_1,\lambda_2}(\min\{d,T\})\)
for probabilistic automata. Four unitary and four permutation transition maps
attain the corresponding orders. At zero false-acceptance tolerance
(\(\lambda_2=0\) with \(0\leq\lambda_1<1\)), all four model classes satisfy
\(\ln N=\Theta(\min\{\ln d,T\})\). The upper bounds follow from state-space
packing and word counting, while the lower bounds combine explicit transition
tuples with Grassmannian and positive-rate codebooks whose existence follows
from the probabilistic method. These results establish sharp
dimension--word-length tradeoffs for bounded-word Hankel tables.
\end{abstract}


\end{frontmatter}

\section{Introduction}

\subsection{Problem setting and Hankel formulation}

The central object of this paper is a square acceptance table generated by a
shared finite-dimensional transition system. Row \(i\) represents the binary
test for candidate \(i\), while column \(j\) represents the state reached when
message \(j\) is encoded. The table is required to have a fixed-margin
approximate-identity form: for states
\(\sigma_1,\ldots,\sigma_N\) and binary tests
\(F_1,\ldots,F_N\), let \(\langle F_i,\sigma_j\rangle\) denote the test--state
pairing, given by the trace pairing in the quantum model and by the usual
inner product between a test vector and a probability vector in the classical
model. We require
\begin{equation}
 \begin{aligned}
  \langle F_i,\sigma_i\rangle&\ge 1-\lambda_1,\\
  \langle F_i,\sigma_j\rangle&\le \lambda_2
  &&(i\ne j).
 \end{aligned}
 \label{eq:intro-id-table}
\end{equation}
These are the standard diagonal and off-diagonal acceptance constraints in
identification theory
\cite{ahlswede1989identification,winter2004identification}.

Both the states and the tests are generated by words over a single shared
finite alphabet. Specifically, a prefix word \(u_j\) maps a common initial
state to the state associated with message \(j\), while a suffix word \(v_i\)
propagates a common terminal effect backward to the test associated with
candidate \(i\). If \(h\) denotes the resulting acceptance series, the
\((i,j)\)-entry of the response table is \(h(u_jv_i)\). Under the standard
Hankel convention \(H_h(u,v)=h(uv)\), this response table is the transpose of
the selected prefix--suffix submatrix. Thus, the problem is to determine the
largest admissible side length \(N\) of a square Hankel submatrix satisfying
the prescribed margins when every prefix and suffix has length at most \(T\).
The formal realization classes and the corresponding optimization functional
are introduced in Section~\ref{sec:model-results}.

In the classical setting, the state space is the probability simplex on a
\(d\)-point carrier, whereas in the quantum setting it is the set of density
operators on a \(d\)-dimensional Hilbert space. We consider a four-symbol
alphabet whose transition maps are shared by all selected rows and columns.
The two word-length budgets are imposed separately: the encoder and the
tester each receive a budget \(T\), so the concatenated Hankel word
\(u_jv_i\) has length at most \(2T\).

\subsection{Relation to prior work}

Our framework is grounded in weighted-automata and finite-dimensional
realization theory. Probabilistic automata, finite-rank Hankel
representations, and measure-once quantum finite automata provide standard
classical and quantum models for acceptance series
\cite{rabin1963probabilistic,carlyle1971realizations,
moore2000quantum,ambainis1998qfa}. Fooling-set arguments analyze related
identity-like concatenation patterns
\cite{glaister1996lower}. The optimization studied here differs in direction
from Hankel-rank and state-complexity problems: rather than fixing a target
behavior and minimizing the realization dimension, we fix the state dimension
\(d\), word-length bound \(T\), alphabet budget \(q\), and acceptance margins,
and maximize the order of an admissible square Hankel submatrix. The
distinction is summarized in Table~\ref{tab:optimization-comparison}.

\begin{table}[t]
\centering
\small
\setlength{\tabcolsep}{5pt}
\renewcommand{\arraystretch}{1.25}
\caption{Optimization directions in related realization problems.}
\label{tab:optimization-comparison}
\begin{tabular}{>{\raggedright\arraybackslash}p{0.26\textwidth}
                 >{\raggedright\arraybackslash}p{0.31\textwidth}
                 >{\raggedright\arraybackslash}p{0.35\textwidth}}
\hline
Problem & Fixed input & Optimized quantity \\
\hline
Hankel rank / realization
& Acceptance series \(h\)
& Minimum realization dimension \\

State complexity / simulation
& Language, cutpoint behavior, or automaton computation
& Number of states needed to reproduce the behavior \\

Bounded-word Hankel submatrices (this work)
& State dimension \(d\), word-length bound \(T\), acceptance margins, and
alphabet budget \(q\)
& Maximum order of a square Hankel submatrix satisfying the margins \\
\hline
\end{tabular}
\end{table}

Classical and quantum channel identification provide the static code-size
benchmarks for the bounded-word problem considered here. For fixed acceptance
margins, the corresponding logarithmic code-size orders are
\(\Theta(d)\) in the classical setting and \(\Theta(d^2)\) in the quantum
setting
\cite{ahlswede1989identification,winter2004identification}. Our formulation
replaces unconstrained choices of code states and tests by bounded words in a
single shared transition system. The quantum construction attains the
\(\Theta(d^2)\) order in the mixed-state regime, whereas the pure-state
fingerprinting regime has \(\ln N=\Theta(d)\)
\cite{winter2004identification,buhrman2001fingerprinting}.

A related realization problem minimizes the state dimension needed to
reproduce a prescribed finite family of initial-state--word responses under
shared stochastic maps and a single terminal effect
\cite{zhao2026scrc}. In the present work, all message states are generated
from one common initial state, the selected states and tests form a complete
approximate-identity table, and the objective is to maximize the table order
under simultaneous dimension and word-length budgets.

Recent automata-theoretic work establishes tight state-complexity laws for
four-letter simulations and strict-cutpoint behavior
\cite{chenwu2026statecost,chen2026simulation,chen2026behavioral}. These
results fix a target language, cutpoint behavior, or automaton and quantify
the number of states required to reproduce it. They therefore address a
different optimization direction from the one studied here, which maximizes
the largest fixed-margin Hankel submatrix while optimizing the transition
system itself. Complementary work on deterministic identification for
memoryless channels and identification over permutation channels develops
output-space and block-channel viewpoints
\cite{colomer2025dimensional,colomer2025quantumHT,
sarkar2025permutation}. In our model, by contrast, \(T\) counts sequential
transitions acting on a single persistent state space of dimension \(d\).

The quantum lower bound uses Grassmannian packing to construct a large
fixed-margin Hankel submatrix. Projection packings and their volume bounds
are classical
\cite{conway1996grassmannian,barg2002grassmann}, while reflection products
and finite-group orbits provide related geometric constructions
\cite{pitaval2012grassmannianorbits,stewart1980randomorthogonal,
porod1996randomreflections,irving2012banded}. In our construction, four
explicit shared transition maps generate the reachable message states, and
exact inverse words generate the matching backward-propagated tests. Related
overlap and orbit techniques also arise in quantum fingerprinting and
quantum recurrence studies
\cite{buhrman2001fingerprinting,ambainis2009improved,
vasiliev2016hashing,kotowski2026recurrence,gupta2026recurrence}.

\subsection{Contributions and organization}

The primary contributions of this work are as follows:
\begin{itemize}
  \item We formulate identification as a bounded-word Hankel-submatrix
  problem over finite-dimensional realizations with a shared transition
  alphabet.

  \item We establish universal converse bounds from state-space packing and
  word counting.

  \item We give matching four-transition achievability constructions in the
  reversible unitary and permutation subclasses.

  \item We characterize the sharp boundary at zero false-acceptance
  tolerance.
\end{itemize}

From a computational-theoretic perspective, the core contribution is the
formulation and solution of an outer optimization problem for fixed-margin
Hankel submatrices over finite-dimensional realizations with a shared
transition alphabet. This optimization direction differs from both
Hankel-rank minimization and state-complexity bounds for a fixed target
behavior.

The resulting scaling laws are
\(\Theta_{\lambda_1,\lambda_2}(\min\{d^2,T\})\) in the quantum models and
\(\Theta_{\lambda_1,\lambda_2}(\min\{d,T\})\) in the classical models for every
fixed positive false-acceptance tolerance. At zero false acceptance, all four
model classes satisfy the logarithmic law
\(\Theta(\min\{\ln d,T\})\). The precise definitions, parameter quantifiers,
and threshold ranges are given in Section~\ref{sec:model-results}.
The remainder of the paper is organized as follows. Section~\ref{sec:model-results}
introduces the formal bounded-word realization model, its Hankel
reformulation, and the main theorem statements. Section~\ref{sec:converses}
establishes the universal converse bounds. Sections~\ref{sec:classical} and
\ref{sec:quantum} present the classical permutation and quantum Grassmannian
achievability constructions, respectively. Section~\ref{sec:zero} analyzes
the zero-false-acceptance regime and the resulting polynomial word-length
scaling. Section~\ref{sec:discussion} discusses the resulting
dimension--word-length phase diagram, finite-precision stability, and
structural extensions. Finally, Section~\ref{sec:conclusion} concludes the
paper.

\section{Formal Model and Theorem Statements}
\label{sec:model-results}

\subsection{Finite-dimensional transition model with bounded word length}

Throughout this paper, words are read from left to right, so the leftmost
symbol acts first. In the quantum classes, each input letter \(a\) acts by a CPTP map \(\Phi_a\); in the unitary
subclass, \(\Phi_a(X)=U_a X U_a^*\) for some \(U_a\in U(d)\). In the
stochastic classes, each letter \(a\) acts by a row-stochastic matrix \(P_a\).

For a word \(w=a_1\cdots a_\ell\), sequential composition in the quantum
classes is
\begin{equation}
  \Phi_w=\Phi_{a_\ell}\circ\cdots\circ\Phi_{a_1},
  \qquad
  U_w=U_{a_\ell}\cdots U_{a_1}
  \quad\text{in the unitary subclass}.
  \label{eq:chrono}
\end{equation}
In the stochastic classes, we write
\begin{equation}
  P_w=P_{a_1}\cdots P_{a_\ell},
  \qquad
  pP_w=pP_{a_1}\cdots P_{a_\ell},
  \label{eq:chrono-classical}
\end{equation}
where \(p\) is a row distribution. Consequently,
\(\Phi_{uv}=\Phi_v\circ\Phi_u\) and \(P_{uv}=P_uP_v\).

Given an initial state \(\rho\) and a terminal effect \(E\) in the quantum
model, or an initial distribution \(p\) and a terminal response vector \(f\)
in the classical model, a run with encoding prefix \(u\) and test suffix \(v\)
has acceptance probability
\begin{equation}
 \begin{aligned}
  p(u,v)
  &=
  \operatorname{Tr}\!\left[E\Phi_v(\Phi_u(\rho))\right],
  &&\text{quantum},\\
  p(u,v)
  &=pP_uP_vf,
  &&\text{classical}.
 \end{aligned}
 \label{eq:response}
\end{equation}

The prefix and suffix satisfy the separate length bounds
\[
 |u|\le T,
 \qquad
 |v|\le T.
\]
Thus, the encoder and the verifier each receive a maximum word-length budget
\(T\). If instead a single total budget
\(|u|+|v|\le T\) is imposed, the same asymptotic scaling laws hold up to
absolute constant factors.

\begin{definition}[Maximum admissible square Hankel-submatrix size]
\label{def:identification-number}
Let $d,q\ge1$ and $T\ge0$ be integers, and let
$\lambda_1,\lambda_2\in[0,1]$ satisfy
$\lambda_1+\lambda_2<1$.
Let $\mathsf X$ be one of the following model classes:
\begin{itemize}
\item $\mathsf Q$: CPTP maps on
      $M_d(\mathbb C)$, a density operator
      $\rho$, and an effect $0\le E\le I$;
\item $\mathsf U$: the subclass in which every map is conjugation by
      a unitary in $U(d)$;
\item $\mathsf S$: row-stochastic matrices on a $d$-point carrier,
      a distribution $p$, and a response vector $f\in[0,1]^d$;
\item $\mathsf P$: the subclass in which every stochastic matrix is a
      permutation matrix.
\end{itemize}
Here $d$ denotes the state dimension: the Hilbert-space
dimension in the quantum classes and the carrier cardinality in the classical classes.
For an alphabet budget $q$, let
$N_{\mathsf X,q}(d,T;\lambda_1,\lambda_2)$ be the largest $N$ for
which there are at most $q$ shared transition maps, one initial state, one
terminal effect, encoding words $u_1,\ldots,u_N$, and suffix words
$v_1,\ldots,v_N$, with $|u_j|,|v_i|\le T$. We use
$N_{\mathsf X,q}$ for this maximum. With the
first index always denoting the test and the second the message, define
\begin{equation}
 p_{ij}:=
 \begin{cases}
  \Tr[E\Phi_{v_i}(\Phi_{u_j}(\rho))],
      &\mathsf X\in\{\mathsf Q,\mathsf U\},\\
  pP_{u_j}P_{v_i}f,
      &\mathsf X\in\{\mathsf S,\mathsf P\}.
 \end{cases}                                      \label{eq:pij}
\end{equation}
The code is admissible when every one of these $N^2$ responses satisfies
\begin{equation}
 p_{ij}\ge 1-\lambda_1\quad(i=j),
 \qquad
 p_{ij}\le \lambda_2\quad(i\ne j).              \label{eq:IDtable}
\end{equation}
Thus row $i$ specifies a binary match test for candidate $i$ with acceptance
and rejection thresholds given by the two margins.
The initial state, terminal effect, and transitions are shared across all pairs
$(i,j)$.  For each parameter tuple $(d,T,\lambda_1,\lambda_2)$, the
maximization is taken jointly over the initial state, terminal effect,
transition tuple, and both word families. Thus $N_{\mathsf X,q}$ is the
maximum admissible square-submatrix size at state
dimension \(d\), maximum word length \(T\), and alphabet size \(q\).
\end{definition}
\begin{samepage}
For $T\ge0$, write $\mathcal A^{\le T}$ for the words of length at most $T$. For an acceptance series $h$ (the formal power series assigning an acceptance value to each word), let $H_h^{(T)}:=(h(uv))_{u,v\in\mathcal A^{\le T}}$ be the finite restriction (truncation) of the Hankel matrix; its rows and columns are indexed by $\mathcal A^{\le T}$. Define
\(
 \operatorname{M}_{\lambda_1,\lambda_2}(H_h^{(T)})
  :=\max\left\{N:\begin{array}{l}
  U=(u_1,\ldots,u_N)\in(\mathcal A^{\le T})^N,\\
  V=(v_1,\ldots,v_N)\in(\mathcal A^{\le T})^N,\\
  h(u_iv_i)\ge 1-\lambda_1\text{ for all }i,\\
  h(u_jv_i)\le \lambda_2\text{ for all }i\ne j
  \end{array}\right\}.
\)
The maximization is over ordered word lists; when $\lambda_1+\lambda_2<1$, the margins force all selected prefixes and suffixes to be distinct.
Let $\mathcal R_{\mathsf X,q,d}$ be the union over all alphabets $\mathcal A$ with $|\mathcal A|\le q$ and all admissible shared realizations of state dimension $d$. For $\mathsf X\in\{\mathsf Q,\mathsf U\}$, the realization tuple is
$(\rho,E,\{\Phi_a\}_{a\in\mathcal A})$, with $h(w)=\operatorname{Tr}[E\Phi_w(\rho)]$; for $\mathsf X\in\{\mathsf S,\mathsf P\}$, it is
$(p,f,\{P_a\}_{a\in\mathcal A})$, with $h(w)=pP_wf$.
\end{samepage}

\begin{proposition}[Hankel-submatrix reformulation]
\label{prop:hankel-reformulation}
For each class $\mathsf X$, alphabet budget $q$, and parameter tuple
$(d,T,\lambda_1,\lambda_2)$, one has
\(
 N_{\mathsf X,q}(d,T;\lambda_1,\lambda_2)
 =\sup_{h\in\mathcal R_{\mathsf X,q,d}}
   \operatorname{M}_{\lambda_1,\lambda_2}(H_h^{(T)}).
\)
Here the order of a square table is its side length $N$. Thus the
quantity $N_{\mathsf X,q}$ is the largest side length of a square Hankel
submatrix indexed by prefix and suffix words of length at most $T$,
where the test index $i$ labels rows and the message index $j$ labels
columns. Under the standard Hankel orientation $H_h(u,v)=h(uv)$, the matrix
$M$ is the transpose of the selected Hankel submatrix: the prefixes $u_j$ are
its rows and the suffixes $v_i$ are its columns. More formally, for ordered
lists $U=(u_1,\ldots,u_N)$ and $V=(v_1,\ldots,v_N)$, define
\(
 [H_h(U,V)]_{j,i}:=h(u_jv_i).
\)
Then $M=[H_h(U,V)]^\top.$
\end{proposition}

\begin{proof}
Given an admissible identification code, define
$h(w)=\operatorname{Tr}[E\Phi_w(\rho)]$ in the quantum classes and
$h(w)=pP_wf$ in the classical classes. The composition convention gives
$p_{ij}=h(u_jv_i)$, so the response matrix is the transpose of the
corresponding standard Hankel submatrix and satisfies \eqref{eq:IDtable}.
Conversely, any such linear realization of an acceptance series
provides the initial state, terminal effect, transition maps, and the two word families
required in Definition~\ref{def:identification-number}; the same inequalities make
it an admissible identification code. The two optimizations therefore have
the same value.
\end{proof}

\begin{proposition}[Hankel-rank bounds of the acceptance series]
\label{prop:hankel-rank}
Let \(H_h=(h(uv))_{u,v\in\mathcal A^*}\) be the acceptance Hankel matrix
associated with a class-\(\mathsf X\) realization of state dimension \(d\).
Then, over the real numbers,
\(
\operatorname{rank}H_h\le d
\quad\text{for }\mathsf X\in\{\mathsf S,\mathsf P\},
\qquad
\operatorname{rank}H_h\le d^2
\quad\text{for }\mathsf X\in\{\mathsf Q,\mathsf U\}.
\)
\end{proposition}

\begin{proof}
For a stochastic realization, write \(h(uv)=pP_uP_vf\). This is a
finite-dimensional linear representation through the \(d\)-dimensional row
vectors \(pP_u\), so every finite Hankel submatrix has rank at most \(d\).
For a quantum realization, restrict the channels to the real vector space of
Hermitian \(d\times d\) matrices, which has dimension \(d^2\). Each
\(\Phi_a\) induces a real linear map on this space, and
\(h(uv)=\langle E,\Phi_v\Phi_u(\rho)\rangle\) is a factorization through that
\(d^2\)-dimensional space. Hence every finite Hankel submatrix has rank at
most \(d^2\).
\end{proof}

The inclusions of model classes imply
\begin{equation}
 N_{\mathsf U,q}\le N_{\mathsf Q,q},
 \qquad
 N_{\mathsf P,q}\le N_{\mathsf S,q}.            \label{eq:inclusions}
\end{equation}

\subsection{Main results}

\begin{theorem}[Dimension--word-length tradeoff for four transitions]
\label{thm:main}
Fix \(\lambda_1\ge0\), \(\lambda_2>0\), and
\(\lambda_1+\lambda_2<1\). There exist positive constants
\(c,C,R_0\), depending only on \(\lambda_1,\lambda_2\), such that, with
\(N_{\mathsf X}=N_{\mathsf X,4}(d,T;\lambda_1,\lambda_2)\),
\begin{align}
 c\min\{d^2,T\}
 &\le \ln N_{\mathsf U}
 \le \ln N_{\mathsf Q}
 \le C\min\{d^2,T\},                              \label{eq:mainQ}\\
 c\min\{d,T\}
 &\le \ln N_{\mathsf P}
 \le \ln N_{\mathsf S}
 \le C\min\{d,T\}.                               \label{eq:mainC}
\end{align}
The first chain holds whenever \(\min\{d^2,T\}\ge R_0\), and the second
whenever \(\min\{d,T\}\ge R_0\). The bounds are uniform over all integer
pairs \(d\ge1\), \(T\ge0\) in the corresponding threshold ranges. For each
parameter tuple, the four transition maps, initial state, and terminal effect
may depend on \(d,T,\lambda_1,\lambda_2\), while the alphabet cardinality
remains four. The lower bounds are attained by four real orthogonal
transitions in the quantum subclass and four permutations in the classical
subclass, both with unit diagonal acceptance. The result establishes that
four maps suffice; it does not assert that four is the minimum alphabet size.
All logarithms are natural.
\end{theorem}
The constants are uniform for each fixed $\lambda_2>0$; at $\lambda_2=0$,
support separation gives a distinct scaling law.

\begin{remark}[Interpretation of the four-transition result]
The theorem establishes \(q=4\) as a sufficient four-symbol alphabet, with transition maps optimized for each parameter tuple and fixed error margins. The four maps are explicit, while probabilistic selection produces exponential-size word families for both lower bounds; these combinatorial codebooks establish the existence result. Efficient indexing and deterministic extraction remain open algorithmic questions. The same lower bounds hold for every $q\ge4$ by padding with unused letters. The fixed-generator, uniform-in-dimension problem and the minimization of the alphabet cardinality are separate structural extensions.
\end{remark}

\paragraph{Optimization result and relation to Hankel rank}
The main theorem is an optimization result over the acceptance series, the
shared transition maps, the initial state, the terminal functional, and
bounded prefix--suffix word lists. It determines the optimal order of a square
Hankel submatrix satisfying fixed diagonal and off-diagonal margins under
simultaneous state-dimension and word-length constraints. The theorem
complements the Hankel-rank bound by optimizing the realization and
determining the sharp submatrix-size scaling, with matching reversible
constructions.

\begin{corollary}[Scaling law for four-symbol Hankel submatrices]
\label{cor:hankel-law}
Under the hypotheses of Theorem~\ref{thm:main}, let
\(\mathcal H_{\mathsf X,4}(d,T;\lambda_1,\lambda_2)\) denote the largest
order of a square Hankel submatrix satisfying the fixed margins in a
class-\(\mathsf X\) acceptance series, with row and column words of length at
most \(T\). By Proposition~\ref{prop:hankel-reformulation},
\(\mathcal H_{\mathsf X,4}=N_{\mathsf X,4}\). In the corresponding threshold
ranges,
\begin{align}
\ln \mathcal H_{\mathsf Q,4}
 &=\Theta_{\lambda_1,\lambda_2}\!\left(\min\{d^2,T\}\right), \label{eq:corQ}\\
\ln \mathcal H_{\mathsf S,4}
 &=\Theta_{\lambda_1,\lambda_2}\!\left(\min\{d,T\}\right). \label{eq:corS}
\end{align}
The same orders hold for the unitary and permutation subclasses. Moreover,
the ambient acceptance Hankel matrices have ranks at most \(d^2\) in the
quantum classes and at most \(d\) in the classical classes.
\end{corollary}

The rank bounds in Proposition~\ref{prop:hankel-rank} and the size laws in
Corollary~\ref{cor:hankel-law} describe different aspects of the same table.
Rank controls the dimension of a linear realization of the acceptance
series, while the order of a fixed-margin square submatrix is governed by the
state-space packing and word-length constraints. Suitable rank-$D$ Gram matrices can contain an exponentially large
fixed-tolerance square submatrix when their off-diagonal entries are
allowed a fixed positive tolerance. The exponential orders in
Corollary~\ref{cor:hankel-law} therefore come from the state-space packing
and bounded-length word constructions, while the rank proposition identifies the
ambient realization dimension in which those tables are embedded.

\begin{remark}[Computational content of the Hankel reformulation]
For a fixed acceptance series, Proposition~\ref{prop:hankel-rank} is a
realization-dimension statement. Theorem~\ref{thm:main} addresses a different
outer optimization: it ranges over the shared transition tuple and the
initial and terminal functionals, then measures the growth of the largest
square Hankel submatrix satisfying the prescribed diagonal and off-diagonal margins inside the corresponding
prefix and suffix word sets of length at most $T$. Thus the main theorem is an
optimal scaling law for low-dimensional realizations with bounded word length;
rank
minimization and automaton simulation cost are separate optimization problems.
\end{remark}

Figure~\ref{fig:phase-diagram} summarizes the resulting maximum-word-length regimes and
the sharp change at zero off-diagonal acceptance.

\begin{figure*}[t]
\centering
\includegraphics[width=\textwidth]{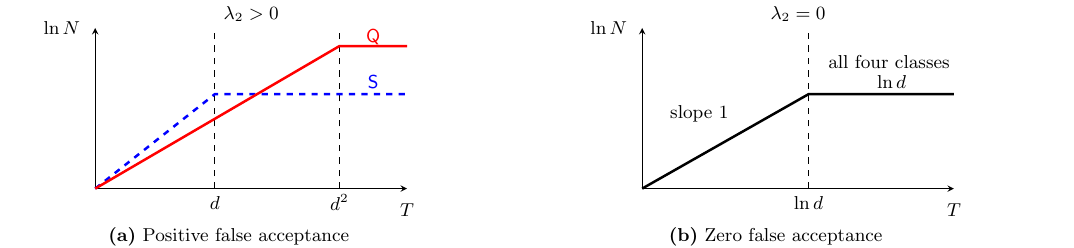}
\caption{Schematic dimension--word-length phase diagram.  In panel (a),
${\mathsf Q}$ denotes the quantum/CPTP model and ${\mathsf S}$ denotes the
classical/stochastic model. For fixed positive off-diagonal tolerance,
the solid red curve is the quantum model and the dashed blue curve is the
classical model; their laws are respectively
$\Theta(\min\{d^2,T\})$ and $\Theta(\min\{d,T\})$. At zero
off-diagonal tolerance, all four model classes obey
$\Theta(\min\{\ln d,T\})$.  The curves depict the asymptotic scaling orders;
finite-dimensional constants are suppressed.}
\label{fig:phase-diagram}
\end{figure*}

\begin{theorem}[Zero false-acceptance case]               \label{thm:zero}
For every pair of integers $d\ge1$, $T\ge0$, every
$0\le\lambda_1<1$, and
$\mathsf X\in\{\mathsf P,\mathsf S,\mathsf U,\mathsf Q\}$,
\begin{equation}
 \begin{split}
 \min\{d,2^{T+1}-1\}
 &\le N_{\mathsf X,4}(d,T;\lambda_1,0)\\
 &\le\min\left\{d,\frac{4^{T+1}-1}{3}\right\}.
 \end{split}                                      \label{eq:zerobounds}
\end{equation}
More explicitly, there are absolute constants $c_0,C_0>0$ such that,
uniformly for all $d\ge3$, $T\ge1$, $0\le\lambda_1<1$, and all four
model classes,
\begin{equation}
 c_0\min\{\ln d,T\}
 \le \ln N_{\mathsf X,4}(d,T;\lambda_1,0)
 \le C_0\min\{\ln d,T\}.                       \label{eq:zerolaw}
\end{equation}
\end{theorem}
The lower bound is realized by a binary prefix tree whose two partial transition
maps extend to permutations of the carrier; inverse letters provide the remaining two
symbols. The corresponding permutation unitaries give the unitary subclass, and
the stochastic and CPTP classes contain these constructions.

\begin{corollary}[Double-logarithmic growth for polynomial word length]
\label{cor:poly-horizon}
Fix errors as in Theorem~\ref{thm:main}, and let $(T_d)_{d\ge1}$ be a
sequence of nonnegative integers satisfying
\begin{equation}
 T_d=d^{\tau+o(1)},\qquad \tau>0.                  \label{eq:polyhorizon}
\end{equation}
For $\mathsf X\in\{\mathsf Q,\mathsf U,\mathsf S,\mathsf P\}$, write
\begin{equation}
 \gamma_{\mathsf X}:=
 \lim_{d\to\infty}
 \frac{\ln\ln N_{\mathsf X,4}
 (d,T_d;\lambda_1,\lambda_2)}{\ln d}.
 \label{eq:polydef}
\end{equation}
Then these limits exist and
\begin{align}
 \gamma_{\mathsf Q}=\gamma_{\mathsf U}&=\min\{\tau,2\},
                                                        \label{eq:polyQ}\\
 \gamma_{\mathsf S}=\gamma_{\mathsf P}&=\min\{\tau,1\}.
                                                        \label{eq:polyC}
\end{align}
In particular, every $\tau>1$ yields a strict quantum--classical
exponent separation under the same four-letter budget.
\end{corollary}

\section{Universal converses}
\label{sec:converses}

For $q\ge2$, put
\begin{equation}
 \cW_q(T)=\sum_{t=0}^Tq^t=\frac{q^{T+1}-1}{q-1}.  \label{eq:wordcount}
\end{equation}

\begin{lemma}[Packing in a normed affine space]    \label{lem:volume}
Let $x_1,\ldots,x_N$ lie in a $D$-dimensional affine space equipped
with a norm.  If its set of admissible points has diameter at most $2$
and $\lVert x_i-x_j\rVert\ge2\delta$ for $i\ne j$, where
$\delta>0$, then
\begin{equation}
 N\le\left(1+\frac2\delta\right)^D.               \label{eq:volbound}
\end{equation}
\end{lemma}

\begin{proof}
Translate the affine space so that $x_1$ is the origin.  The open norm
balls of radius $\delta$ about the $x_i$'s are disjoint, and every
such ball is contained in the ball of radius $2+\delta$ about
$x_1$.  Translation invariance and homogeneity of $D$-dimensional
Lebesgue measure give
$N\delta^D\le(2+\delta)^D$, proving
\eqref{eq:volbound}.
\end{proof}

\begin{proposition}[State packing and word counting] \label{prop:converse}
Let $\delta=1-\lambda_1-\lambda_2>0$.  For every $q\ge2$,
write
$N_{\mathsf X}=N_{\mathsf X,q}(d,T;\lambda_1,\lambda_2)$.  Then
\begin{align}
 \ln N_{\mathsf Q}
 &\le \min\left\{
 (d^2-1)\ln\left(1+\frac2\delta\right),
 \ln\cW_q(T)\right\},                            \label{eq:Qconv}\\
 \ln N_{\mathsf S}
 &\le \min\left\{
 (d-1)\ln\left(1+\frac2\delta\right),
 \ln\cW_q(T)\right\}.                            \label{eq:Sconv}
\end{align}
The same bounds hold for the corresponding subclasses
$\mathsf U$ and $\mathsf P$.
\end{proposition}

\begin{proof}
In the quantum model let
$\rho_j=\Phi_{u_j}(\rho)$ and pull the terminal effect back along suffix $v_i$ to
$F_i=\Phi_{v_i}^*(E)$.  Complete positivity and unitality of the
adjoint give $0\le F_i\le I$.  For $j\ne k$, test $j$ and
\eqref{eq:IDtable} give
\begin{equation}
 \Tr[F_j(\rho_j-\rho_k)]\ge\delta.
\end{equation}
The variational characterization of trace distance implies
\begin{equation}
 \frac12\lVert\rho_j-\rho_k\rVert_1\ge\delta.      \label{eq:trsep}
\end{equation}
Density operators form a subset of diameter at most $2$ in the trace
norm inside the real affine space of trace-one Hermitian matrices, whose
dimension is $d^2-1$.  Lemma~\ref{lem:volume} proves the first term of
\eqref{eq:Qconv}.

Classically, write $x_j=pP_{u_j}$ and
$g_i=P_{v_i}f\in[0,1]^d$.  Then
$(x_j-x_k)g_j\ge\delta$, and the variational formula for total
variation gives
\begin{equation}
 \frac12\lVert x_j-x_k\rVert_1\ge\delta.
\end{equation}
The probability simplex has affine dimension $d-1$ and $\ell_1$
diameter $2$, so Lemma~\ref{lem:volume} gives the first term of
\eqref{eq:Sconv}.

The suffix count follows directly from word uniqueness.  If $v_i=v_k$ for distinct $i,k$, then
on prefix $u_i$ the same response is at least $1-\lambda_1$, because
test $i$ matches, and at most $\lambda_2$, because test $k$ is
nonmatching.  The gap $\delta>0$ makes these bounds incompatible.  Hence
the $N$ suffix words are
syntactically distinct.  The total number of words of length at most $T$ is
$\cW_q(T)$, which proves the second term in both bounds.  The subclass
bounds follow from \eqref{eq:inclusions}.
\end{proof}

\section{Classical construction of a Hankel submatrix with bounded word length}
\label{sec:classical}

The classical construction uses two cyclic permutations on a carrier partitioned into blocks:
one moves through the block addresses and the other applies a symbol offset at
the current address. An encoded word therefore stores one code symbol in each
block. The inverse permutation word at the verifier removes the candidate
offset, and the common terminal indicator test accepts according to normalized
coordinate agreement of the two codewords. A positive-rate \(q\)-ary code with a
fixed relative Hamming-distance bound supplies the required off-diagonal margin.

\begin{lemma}[A positive-rate q-ary code with prescribed relative distance]       \label{lem:agreement}
Fix an integer $Q_0\ge2$ and a number
$Q_0^{-1}<\beta<1$.  There are constants
$\kappa>0$ and $m_0<\infty$ such that, for every $m\ge m_0$,
there is a code $\cC_m\subseteq\mathbb Z_{Q_0}^m$ with
\begin{equation}
 |\cC_m|\ge e^{\kappa m}
 \quad\text{and}\quad
  \frac1m|\{b:c_b\ne c_b'\}|>1-\beta
 \quad(c\ne c').                                  \label{eq:agreement}
\end{equation}
\end{lemma}

\begin{proof}
For independent uniform $c,c'$, their agreement count is
$X\sim\mathrm{Bin}(m,Q_0^{-1})$.  Exponential Markov inequality,
optimized over its positive parameter, gives
\begin{equation}
 \Pr\{X\ge\beta m\}
 \le \exp[-mD(\beta\Vert Q_0^{-1})],              \label{eq:chernoff}
\end{equation}
where
$D(a\Vert b)=a\ln(a/b)+(1-a)\ln((1-a)/(1-b))$.
Indeed, minimizing
$e^{-t\beta m}(1-Q_0^{-1}+Q_0^{-1}e^t)^m$ over $t>0$
is exactly \eqref{eq:chernoff}.  Let
$D=D(\beta\Vert Q_0^{-1})>0$. Draw
$L=\lfloor e^{Dm/4}\rfloor$ independent words. For sufficiently large $m$,
$L\ge e^{Dm/8}$, so one may take $\kappa=D/8$ after increasing $m_0$.
A union bound over fewer than $L^2/2$ pairs gives bad probability at most
$\tfrac12e^{-Dm/2}$, which is below one for all sufficiently large $m$.
\end{proof}

\begin{proposition}[Permutation construction]           \label{prop:classical}
Under the hypotheses of Theorem~\ref{thm:main}, there are constants
$c_{\lambda_2}>0$ and $R_{\lambda_2}<\infty$ such that
\begin{equation}
 \ln N_{\mathsf P,4}(d,T;\lambda_1,\lambda_2)
 \ge c_{\lambda_2}\min\{d,T\}                    \label{eq:Plower}
\end{equation}
whenever $\min\{d,T\}\ge R_{\lambda_2}$.
\end{proposition}

\begin{proof}
Choose an integer $Q_0>2/\lambda_2$ and set
$\beta=(\lambda_2+Q_0^{-1})/2$. Then $Q_0^{-1}<\beta<\lambda_2$, and both
$Q_0$ and $\beta$ depend only on $\lambda_2$.
For an integer $m$, use the carrier
$\cX=\mathbb Z_m\times\mathbb Z_{Q_0}$, and put
\begin{equation}
 S_0=\{(b,0):b\in\mathbb Z_m\}.
\end{equation}
The initial distribution is uniform on \(S_0\), and the final output function is
$f=\one_{S_0}$. Use the four-symbol input alphabet
$C,C^{-1},R,R^{-1}$, acting by the permutations
\begin{align}
 C(b,z)&=(b-1\!\pmod m,z),                         \label{eq:Cperm}\\
 R(b,z)&=
 \begin{cases}
  (0,z+1\!\pmod {Q_0}),&b=0,\\
  (b,z),&b\ne0.
 \end{cases}                                      \label{eq:Rperm}
\end{align}

For $c=(c_0,\ldots,c_{m-1})\in\mathbb Z_{Q_0}^m$, with each
symbol represented in $\{0,\ldots,Q_0-1\}$, define the
prefix word
\begin{equation}
 s(c)=R^{c_0}C\,R^{c_1}C\cdots R^{c_{m-1}}C.       \label{eq:scanword}
\end{equation}
Immediately before the $k$-th displayed $R$-block, a point that
started in block $b$ lies in block $b-k\pmod m$.  It is therefore
changed by $R$ precisely at stage $k=b$.  After all $m$ copies of
$C$, its first coordinate has returned to $b$.  Consequently the
permutation $g_c$ induced by \eqref{eq:scanword} satisfies
\begin{align}
 g_c(b,z)&=(b,z+c_b\!\pmod {Q_0}),\notag\\
 g_c(S_0)&=S_c:=\{(b,c_b):b\in\mathbb Z_m\}.        \label{eq:Sc}
\end{align}
Moreover,
\begin{equation}
 |s(c)|=m+\sum_{b=0}^{m-1}c_b\le mQ_0.            \label{eq:scanlength}
\end{equation}
Let $t(c)=s(c)^{-1}$ be the inverse word: reverse the
letters and replace each by its inverse. In particular,
\begin{equation}
 t(c)=C^{-1}R^{-c_{m-1}}\cdots C^{-1}R^{-c_0},
 \qquad |t(c)|=|s(c)|,                             \label{eq:inverseclassical}
\end{equation}
and $g_{t(c)}=g_c^{-1}$.

Take the code $\cC_m$ from Lemma~\ref{lem:agreement}.  Encode
message $j$ by $s(c_j)$ and test candidate $i$ by $t(c_i)$.
For every ordered pair $i,j$, the response equals
\begin{align}
 p_{ij}
 &=\frac1m|\{b:g_{c_i}^{-1}(b,c_{j,b})\in S_0\}| \notag\\
 &=\frac1m|\{b:c_{i,b}=c_{j,b}\}|.                \label{eq:classicaltable}
\end{align}
Thus $p_{ii}=1$, whereas $p_{ij}<\beta<\lambda_2$ whenever
$i\ne j$. Hence the words form an admissible \(N\)-row table under
the shared initial distribution, final output function, and four permutations.

Finally, set
\begin{equation}
 x=\min\{d,T\},\qquad m=\left\lfloor\frac{x}{Q_0}\right\rfloor.
\end{equation}
The working carrier has $mQ_0\le d$ points, and
\eqref{eq:scanlength} bounds both word lengths by $T$. Embed it into a
$d$-point carrier by extending all four permutations by the identity and
extending the initial distribution and response vector by zero on the added
points.  Lemma~\ref{lem:agreement} gives
$\ln N\ge\kappa m$, which implies \eqref{eq:Plower} above a
constant threshold.
\end{proof}

\section{Quantum construction of a Hankel submatrix with bounded word length}
\label{sec:quantum}

The quantum construction is the geometric counterpart of the cyclic
permutation construction for blockwise encoding. A
Householder reflection moves a rank-$r$ subspace while preserving its
dimension, so a sequence of reflections can explore the Grassmann manifold
while preserving unit diagonal acceptance. The identity
$H_s=D_sHD_s$ means that the shift letter $S$ and the sign letter $Z$
realize reflections in signed coordinate directions from a four-symbol
alphabet closed under inverses. Independent sign choices produce products of
reflections; their mean overlap decays and concentration on the Boolean cube
selects an exponentially large family with small pairwise overlap.  Finally,
the verifier uses the inverse word associated with the encoding word. Because
the transitions are unitary, this inverse suffix returns the encoded state to the root,
giving diagonal acceptance one, while
the Grassmann packing controls the off-diagonal entries of the fixed-margin
response table.

\begin{proposition}[Unitary construction]              \label{prop:quantum}
Under the hypotheses of Theorem~\ref{thm:main}, there are constants
$c_{\lambda_2}>0$ and $R_{\lambda_2}<\infty$ such that
\begin{equation}
 \ln N_{\mathsf U,4}(d,T;\lambda_1,\lambda_2)
 \ge c_{\lambda_2}\min\{d^2,T\}                  \label{eq:Ulower}
\end{equation}
whenever $\min\{d^2,T\}\ge R_{\lambda_2}$.
\end{proposition}

\begin{lemma}[Hankel entries generated by inverse words]
\label{lem:quantum-hankel-entry}
Let \(P\) be the rank-\(r\) projection used in the quantum construction,
let \(\rho=P/r\), \(E=P\), and let \(w_i\) be encoding sequences with
unitary products \(U_i\). If \(t_i=w_i^{-1}\) is the inverse word obtained by
reversing \(w_i\) and replacing each symbol by the symbol denoting its inverse,
then the associated acceptance series satisfies
\(
h(w_jt_i)
 =\frac1r\operatorname{Tr}
   \bigl(U_iPU_i^*\,U_jPU_j^*\bigr).
\)
In particular, \(h(w_i t_i)=1\), and the off-diagonal entries are exactly
the normalized Grassmann overlaps.
\end{lemma}

\begin{proof}
By the composition convention,
\(
h(w_jt_i)
 =\operatorname{Tr}\!\left[
 P\,U_i^*U_j\frac{P}{r}U_j^*U_i
\right].
\)
Cyclic invariance of the trace gives the stated expression. For \(i=j\),
the expression equals \(\operatorname{Tr}(P)/r=1\).
\end{proof}

Equivalently, if \(\rho_j=U_jPU_j^*/r\) and \(F_i=U_iPU_i^*\), then
\(
h(w_jt_i)=\operatorname{Tr}(F_i\rho_j).
\)
Thus the inverse suffix words \(t_i\) index the rows of the response table
\((p_{ij})\), while the encoding prefixes \(w_j\) index its columns.
Equivalently, in the standard Hankel orientation \(H_h(u,v)=h(uv)\), the
prefixes \(w_j\) index rows and the inverse words \(t_i\) index columns;
the response table is the transpose of that selected Hankel submatrix.

Proposition~\ref{prop:quantum} follows from the following Grassmannian
construction with bounded word lengths. It supplies the encoded columns and
rows generated by inverse words, equivalently the standard Hankel rows and
columns after transposition.

The proof has four steps. First, the shift and sign reflections synthesize
every signed Householder reflection using at most $O(d)$ letters.  Second, a
recursion for expected projectors shows that a product of independent reflections drives
the expected normalized overlap toward the ambient Grassmannian overlap scale.
Third, the squared-gradient estimate and concentration on the Boolean cube turn this
mean decay into an exponentially small bad-pair probability.  Finally, a
union bound selects an exponentially large word family, and the inverse
words convert it into an identification code with exact diagonal response.
\begin{proposition}[Grassmannian construction of a Hankel submatrix with bounded word length]
                                                        \label{prop:grassmann}
Fix $0<\lambda<1$.  There are constants
$c_\lambda,C_\lambda,d_\lambda>0$ such that, for every
$d\ge d_\lambda$, one can choose
\begin{itemize}
\item an inverse-closed alphabet $\cA_d\subset O(d)$ of exactly four
      matrices;
\item a rank-$r$ projection $P$, with $r=\Theta_\lambda(d)$;
\item $M\ge e^{c_\lambda d^2}$ words
      $w_1,\ldots,w_M$, each of length at most
      $C_\lambda d^2$,
\end{itemize}
such that, writing $U_i=U_{w_i}$,
\begin{equation}
 \frac1r\Tr(U_iPU_i^*U_jPU_j^*)\le\lambda
 \qquad(i\ne j).                                  \label{eq:Grasspack}
\end{equation}
\end{proposition}

\subsection{Proof of the Grassmannian submatrix construction}

\begin{proof}[Proof of Proposition~\ref{prop:grassmann}]

\emph{Exact synthesis of sign--Householder reflections.}

Index coordinates by $\mathbb Z_d$, let
\begin{align}
 u&=d^{-1/2}(1,\ldots,1)^\top,
 &H&=I-2uu^\top,\notag\\
 Se_j&=e_{j+1\pmod d},
 &Z&=I-2e_0e_0^\top.                       \label{eq:fourletters}
\end{align}
and take
\begin{equation}
 \cA_d=\{S,S^{-1},Z,H\}.                           \label{eq:alphabet4}
\end{equation}
This alphabet is inverse closed because $Z^2=H^2=I$.

For $s\in\{\pm1\}^d$, put
$D_s=\diag(s_0,\ldots,s_{d-1})$,
$v_s=D_su=d^{-1/2}s$, and
$\epsilon_j=(1-s_j)/2$.  Since
\begin{equation}
 D_s=\prod_{j=0}^{d-1}S^jZ^{\epsilon_j}S^{-j},     \label{eq:Dsprod}
\end{equation}
ordering the matrix factors by increasing $j$ telescopes
\eqref{eq:Dsprod} to
\begin{equation}
 Z^{\epsilon_0}S Z^{\epsilon_1}S\cdots
 S Z^{\epsilon_{d-1}}S,                           \label{eq:Dsscan}
\end{equation}
where $S^{-(d-1)}=S$ was used at the end.  The displayed matrix
product has at most $2d$ letters.  By reversing its displayed letter
sequence, in accordance with \eqref{eq:chrono}, we obtain a
word $d(s)$ with
\begin{equation}
 U_{d(s)}=D_s,\qquad |d(s)|\le2d.                  \label{eq:Dsword}
\end{equation}
The concatenated word $d(s)H d(s)$ therefore implements exactly
\begin{equation}
 H_s:=I-2v_sv_s^\top=D_sHD_s               \label{eq:Hs}
\end{equation}
and has length at most $4d+1$.

\emph{Recursion for expected projectors.}

Let $P=\diag(I_r,0)$.  For independent uniform sign vectors
$s_1,\ldots,s_m$, set
\begin{align}
 W_0&=I,\qquad
 W_k=H_{s_k}\cdots H_{s_1}\quad(1\le k\le m),\notag\\
 Q_k&=W_kPW_k^\top\quad(0\le k\le m),\notag\\
 F_m&=\frac1r\Tr(PQ_m).                           \label{eq:Fm}
\end{align}
For a deterministic symmetric $Q$, introduce the linear map
\begin{equation}
 \mathcal L(Q)=\mathbb E_s[H_sQH_s].               \label{eq:Ldef}
\end{equation}
The Rademacher fourth-moment identity
\begin{equation}
 \mathbb E[s_is_js_ks_\ell]
 =\delta_{ij}\delta_{k\ell}+\delta_{ik}\delta_{j\ell}
  +\delta_{i\ell}\delta_{jk}
  -2\delta_{ij}\delta_{ik}\delta_{i\ell}           \label{eq:rademacher4}
\end{equation}
follows by checking whether every index occurs an even number of times.
Substitution into each matrix entry gives, for every deterministic
symmetric \(Q\),
\begin{equation}
 \mathbb E_s[v_sv_s^\top Qv_sv_s^\top]
 =\frac{\Tr(Q)I+2Q-2\diag Q}{d^2}.                \label{eq:fourthmoment}
\end{equation}
Expanding $H_sQH_s$, using
$\mathbb E[v_sv_s^\top]=I/d$, and applying
\eqref{eq:fourthmoment} yields
\begin{equation}
 \mathcal L(Q)
 =\left(1-\frac4d\right)Q
  +\frac{4\Tr Q}{d^2}I
  +\frac8{d^2}(Q-\diag Q).                        \label{eq:Lformula}
\end{equation}
In particular, $\mathcal L$ preserves the diagonal subspace and, on
that subspace,
\begin{equation}
 \mathcal L(Q)=\left(1-\frac4d\right)Q
                +\frac{4\Tr Q}{d^2}I.             \label{eq:Ldiag}
\end{equation}
Independence and linearity give
\begin{equation}
 \mathbb E Q_k=\mathcal L(\mathbb E Q_{k-1})
 \qquad(1\le k\le m).                             \label{eq:meanlinear}
\end{equation}
Since $Q_0=P$ is diagonal, $\mathcal L$ preserves the diagonal subspace,
and $\Tr Q_k=r$, induction with \eqref{eq:Ldiag} gives
\begin{equation}
 \mathbb E Q_m
 =\frac rd I+\left(1-\frac4d\right)^m
   \left(P-\frac rd I\right),                     \label{eq:meanQ}
\end{equation}
and hence
\begin{equation}
 \mathbb E F_m
 =\frac rd+\left(1-\frac rd\right)
   \left(1-\frac4d\right)^m.                     \label{eq:meanF}
\end{equation}

\emph{Concentration on the Boolean cube.}

\begin{lemma}[Concentration on the Boolean cube under a bounded discrete gradient]     \label{lem:booleanLSI}
Let $X$ be uniform on $\{\pm1\}^N$, and let $X^{(a)}$ denote
the vector obtained by flipping bit $a$.  If a real function $f$
satisfies, for some $V>0$,
\begin{equation}
 \sup_x\sum_{a=1}^N(f(x)-f(x^{(a)}))^2\le V,       \label{eq:gradV}
\end{equation}
then, for every $t>0$,
\begin{equation}
 \Pr\{f\ge\mathbb Ef+t\}
 \le \exp\left(-\frac{2t^2}{V}\right).           \label{eq:booleanconc}
\end{equation}
\end{lemma}

The proof is given in \ref{app:boolean-concentration}.

\begin{lemma}[Concentration for products of random Householder reflections]
\label{lem:reflection-concentration}
For the random variable $F_m$ in \eqref{eq:Fm} and every $t>0$,
\begin{equation}
 \Pr\{F_m\ge\mathbb EF_m+t\}
 \le\exp\left(-\frac{t^2r^2d}{5120m}\right).     \label{eq:Ftail}
\end{equation}
\end{lemma}

The squared-gradient proof is given in
\ref{app:squared-gradient}.

\emph{Random coding.}

Set
\begin{equation}
 \alpha=\frac\lambda8,\qquad
 r=\lfloor\alpha d\rfloor,\qquad
 n=\left\lceil\frac d8\ln\frac8\lambda\right\rceil. \label{eq:rn}
\end{equation}
Put $L_\lambda=\ln(8/\lambda)$.  Increase $d_\lambda$ so that
$d\ge5$, $r\ge\lambda d/16$, and
$n\le dL_\lambda/2$ for every $d\ge d_\lambda$.  These choices make all
constants below depend only on $\lambda$.
Draw a sampled word
\begin{equation}
 U=H_{s_n}\cdots H_{s_1}                           \label{eq:endpoint}
\end{equation}
from independent uniform signs.  It is implemented using the stated composition convention
concatenating the words for $H_{s_1},\ldots,H_{s_n}$, in this order.

For two independent sampled words $U_i,U_j$, let $s^{(i)}_k$ and $s^{(j)}_k$
denote their sign vectors.  Then
\(
 \begin{aligned}
 U_i^\top U_j
 &=H_{s^{(i)}_1}\cdots H_{s^{(i)}_n}\\[-2pt]
 &\quad H_{s^{(j)}_n}\cdots H_{s^{(j)}_1}.
 \end{aligned}
\)
Every reflection is self-inverse, reversing one i.i.d. sign sequence
preserves its law, and the two sequences are independent.  Thus the relative
product is distributed as a product of $2n$ independent sign reflections.
Cyclicity of trace therefore shows that
\begin{equation}
 G_{ij}:=\frac1r\Tr(U_iPU_i^\top U_jPU_j^\top)
 \quad\text{has the same law as }F_{2n}.           \label{eq:Gijlaw}
\end{equation}
By \eqref{eq:meanF}, $1-x\le e^{-x}$, and \eqref{eq:rn},
\begin{equation}
 \mathbb EG_{ij}
 \le\frac rd+e^{-8n/d}
 \le\frac\lambda8+\frac\lambda8
 =\frac\lambda4.                                  \label{eq:Gmean}
\end{equation}
Taking $t=3\lambda/4$ in \eqref{eq:Ftail}, with $m=2n$, and using
$r\ge\lambda d/16$ and $2n\le dL_\lambda$, define
\begin{equation}
 \gamma_\lambda:=
 \frac{\lambda^4}{2^{24}L_\lambda}>0.                 \label{eq:gammaexplicit}
\end{equation}
Indeed, the exponent in \eqref{eq:Ftail} is bounded below by
\begin{equation}
\frac{(3\lambda/4)^2(\lambda d/16)^2d}
     {5120\,dL_\lambda}
=\frac{9\lambda^4}{5\cdot2^{22}L_\lambda}d^2
\ge\gamma_\lambda d^2.
\end{equation}
Thus, after increasing \(d_\lambda\) only to enforce the preceding
integer-rounding conditions,
\begin{equation}
 \Pr\{G_{ij}>\lambda\}
 \le e^{-\gamma_\lambda d^2}.                    \label{eq:Gtail}
\end{equation}

Draw
$M=\lfloor e^{\gamma_\lambda d^2/4}\rfloor$ independent sampled words.
A union bound gives the explicit estimate
\begin{equation}
 \Pr\{\text{some pair violates \eqref{eq:Grasspack}}\}
 \le \frac{M^2}{2}e^{-\gamma_\lambda d^2}
 \le \frac12e^{-\gamma_\lambda d^2/2}<1                 \label{eq:unionexplicit}
\end{equation}
for all sufficiently large $d$.  Hence a deterministic family with all
pairwise overlaps bounded by $\lambda$ exists.  If $d_\lambda$ is further
increased so that $\gamma_\lambda d^2/4\ge2\ln2$, then
\begin{equation}
 M\ge e^{\gamma_\lambda d^2/8}.                    \label{eq:Mgrass}
\end{equation}
Each sampled word is a concatenation of $n$ words of length at most \(4d+1\).
Since $n\le dL_\lambda/2$ and $d\ge1$,
\begin{equation}
 |w_i|\le n(4d+1)\le \frac52L_\lambda d^2
 \le C_\lambda d^2,                                \label{eq:quadraticlength}
\end{equation}
where one may take $C_\lambda=3L_\lambda$.
Equations \eqref{eq:alphabet4}, \eqref{eq:Grasspack},
\eqref{eq:Mgrass}, and \eqref{eq:quadraticlength} complete the proof.
\end{proof}

\begin{proof}[Proof of Proposition~\ref{prop:quantum}]
First use Proposition~\ref{prop:grassmann} in a working dimension \(m\)
with \(\lambda=\lambda_2\). Set the shared initial state and terminal effect
to
\begin{equation}
 \rho=P/r,\qquad E=P.                              \label{eq:rootPeffect}
\end{equation}
Encode message $j$ with $w_j$.  Since the four-letter alphabet is
inverse closed, test $i$ uses the exact inverse word
\begin{equation}
 w_i^{-1}=a_\ell^{-1}\cdots a_1^{-1}
 \quad\text{when }w_i=a_1\cdots a_\ell.           \label{eq:inverseword}
\end{equation}
By \eqref{eq:chrono}, $U_{w_i^{-1}}=U_i^*$, and the test-$i$,
message-$j$ response is
\begin{align}
 p_{ij}
 &=\Tr\left[P U_i^*U_j\frac PrU_j^*U_i\right] \notag\\
 &=\frac1r\Tr(U_jPU_j^*U_iPU_i^*).                \label{eq:quantumtable}
\end{align}
Thus $p_{ii}=1$, and Proposition~\ref{prop:grassmann} gives
$p_{ij}\le\lambda_2$ for all $i\ne j$.  Hence these words form an
admissible $(M\times M)$ response table, equivalently a transposed Hankel submatrix, under the shared initial state $\rho$,
terminal effect \(E\), and four transitions.

Let $C_{\lambda_2}\ge1$ be the word-length constant in
Proposition~\ref{prop:grassmann}, and choose
\begin{equation}
 m=\left\lfloor
 \min\left\{d,\sqrt{T/C_{\lambda_2}}\right\}
 \right\rfloor.                                   \label{eq:activequantum}
\end{equation}
Increase the resource threshold so that
\(\min\{d,\sqrt{T/C_{\lambda_2}}\}\ge
2\max\{d_{\lambda_2},2\}\). Then \(m\ge d_{\lambda_2}\), and the floor
inequality \(\lfloor x\rfloor\ge x/2\) gives
\begin{equation}
 m^2\ge \frac{1}{4C_{\lambda_2}}
        \min\{d^2,T\}
 =:c'_{\lambda_2}\min\{d^2,T\}.                  \label{eq:mresource}
\end{equation}
Embed the construction into $\mathbb R^d$ by adjoining identity
actions on the remaining $d-m$ coordinates and extending $P$ by zero.
This embedding preserves all word lengths and responses.
Proposition~\ref{prop:grassmann} and
\eqref{eq:mresource} prove \eqref{eq:Ulower}.
\end{proof}

\begin{proof}[Proof of Theorem~\ref{thm:main}]
Proposition~\ref{prop:converse} with $q=4$ gives
\begin{align}
 \ln N_{\mathsf Q,4}
 &\le\min\left\{(d^2-1)\ln(1+2/\delta),
                  \ln\frac{4^{T+1}-1}{3}\right\},\notag\\
 \ln N_{\mathsf S,4}
 &\le\min\left\{(d-1)\ln(1+2/\delta),
                  \ln\frac{4^{T+1}-1}{3}\right\},               \label{eq:finaluppers}
\end{align}
where $\delta=1-\lambda_1-\lambda_2$.  These are respectively
$O_{\lambda_1,\lambda_2}(\min\{d^2,T\})$ and
$O_{\lambda_1,\lambda_2}(\min\{d,T\})$.  Propositions
\ref{prop:quantum} and \ref{prop:classical}, together with
\eqref{eq:inclusions}, supply the matching lower bounds.
\end{proof}

\section{Zero False-Acceptance Case and Polynomial Word-Length Scaling}
\label{sec:zero}

\subsection{Proof of the zero-false-acceptance law}

\begin{proof}[Proof of Theorem~\ref{thm:zero}]
Since $1-\lambda_1>0$, the word-counting argument in
Proposition~\ref{prop:converse}, with $\delta=1-\lambda_1$, gives the
second term in the upper bound. For the first term, pull each suffix back
to an effect.  Quantum mechanically,
$F_i=\Phi_{v_i}^*(E)$ obeys $0\le F_i\le I$,
\begin{equation}
 \Tr(F_i\rho_i)\ge1-\lambda_1>0,
 \qquad
 \Tr(F_i\rho_j)=0\quad(j\ne i).                   \label{eq:exacteffects}
\end{equation}
For positive semidefinite operators, the zero trace in the second
relation implies that $F_i$ annihilates $\supp\rho_j$ whenever $j\ne i$.
The positive diagonal trace permits a choice of
$x_i\in\supp\rho_i$ with $y_i:=F_ix_i\ne0$ and
$\langle y_i,x_i\rangle>0$.  Consequently
\begin{equation}
 \langle y_i,x_j\rangle=0\quad(j\ne i),
 \qquad \langle y_i,x_i\rangle>0.                 \label{eq:biorthogonal}
\end{equation}
Thus the $x_i$ are linearly independent: applying
$\langle y_i,\mathord\cdot\rangle$ to a linear relation isolates its
$i$th coefficient.  Hence $N\le d$.

Classically, write the prefix distributions as $x_j$ and the pulled-back
responses as $g_i\in[0,1]^d$.  Nonnegativity and
$x_jg_i=0$ for $j\ne i$ show that $g_i$ vanishes on
$\supp x_j$.  Since $x_ig_i\ge1-\lambda_1>0$, choose
$a_i\in\supp x_i$ with $g_i(a_i)>0$.  The points $a_i$ are distinct:
if $a_i=a_j$ with $i\ne j$, the off-diagonal condition would give both
$g_i(a_i)>0$ and $g_i(a_j)=0$.  Therefore $N\le d$ also classically.

For the lower bound, put
\begin{equation}
 M=\min\{d,2^{T+1}-1\}.
\end{equation}
Take the first $M$ vertices, in breadth-first order, of the depth-$T$
binary word tree.  They form a prefix-closed set
$W\subseteq\{A,B\}^{\le T}$ of cardinality $M$.  Assign its words
to distinct carrier points $x_w$.  The rules
\begin{equation}
 A(x_w)=x_{wA},\qquad B(x_w)=x_{wB}                \label{eq:treepartial}
\end{equation}
whenever the displayed child belongs to $W$ define two partial
injections.  Each extends to a permutation of the $d$-point carrier,
because its unused domain and unused range have the same cardinality.
Use the four letters $A,B,A^{-1},B^{-1}$.

Starting from $x_\varnothing$, word $w$ reaches
$x_w$.  Write $g_v$ for the permutation induced by word $v$.  Test suffix $v$
uses the exact inverse word $v^{-1}$.  Since the full action is bijective,
\begin{equation}
 \begin{aligned}
 g_v^{-1}(x_w)=x_\varnothing
 &\quad\Longleftrightarrow\quad
 x_w=g_v(x_\varnothing)=x_v\\
 &\quad\Longleftrightarrow\quad w=v.
 \end{aligned}                                    \label{eq:treetable}
\end{equation}
The point mass and indicator at $x_\varnothing$ therefore realize all
$M^2$ cells of the exact identity table.  This is a permutation code
and, by using the corresponding permutation unitaries, also a unitary
code; the stochastic and CPTP cases follow by inclusion.  Taking
logarithms in \eqref{eq:zerobounds} proves \eqref{eq:zerolaw}; indeed,
for $d\ge3$ and $T\ge1$,
\begin{align}
 \ln N_{\mathsf X,4}
 &\ge \min\{\ln d,\ln(2^{T+1}-1)\}
  \ge (\ln2)\min\{\ln d,T\},\notag\\
 \ln N_{\mathsf X,4}
 &\le \min\!\left\{\ln d,
          \ln\frac{4^{T+1}-1}{3}\right\}
  \le (2\ln4)\min\{\ln d,T\}.                   \label{eq:zeroexplicit}
\end{align}
Thus one may take $c_0=\ln2$ and $C_0=2\ln4$, independently of
$\lambda_1$.
\end{proof}

\subsection{Double-logarithmic growth for polynomial word length}

\begin{proof}[Proof of Corollary~\ref{cor:poly-horizon}]
Theorem~\ref{thm:main} and \eqref{eq:polyhorizon} give
\begin{align*}
 \ln N_{\mathsf Q,4},\ \ln N_{\mathsf U,4}
   &=d^{\min\{\tau,2\}+o(1)},\\
 \ln N_{\mathsf S,4},\ \ln N_{\mathsf P,4}
   &=d^{\min\{\tau,1\}+o(1)}.
\end{align*}
Taking one more logarithm and dividing by $\ln d$ proves
\eqref{eq:polyQ}--\eqref{eq:polyC}.
\end{proof}

\section{Consequences, Stability, and Extensions}
\label{sec:discussion}

\subsection{The dimension--word-length phase diagram}

For fixed positive false-acceptance tolerance, the maximum word-length bound has three scaling regimes, up to constants depending on the fixed margins.

\begin{table}[t]
\centering
\small
\setlength{\tabcolsep}{4pt}
\caption{Scaling regimes for the maximum word length.}
\label{tab:regime-table}
\begin{tabular}{p{0.34\textwidth}p{0.25\textwidth}p{0.25\textwidth}}
\hline
Maximum word-length regime & \(\ln N_{\mathsf Q,4}\) & \(\ln N_{\mathsf S,4}\)\\
\hline
\(T\lesssim d\) & \(T\) & \(T\)\\
\(d\lesssim T\lesssim d^2\) & \(T\) & \(d\)\\
\(T\gtrsim d^2\) & \(d^2\) & \(d\)\\
\hline
\end{tabular}
\end{table}
For \(T\lesssim d\), both models are limited by word length and satisfy
\(\ln N=\Theta(T)\).  For \(d\lesssim T\lesssim d^2\), the classical
law has reached its \(\Theta(d)\) state-space upper bound, while the quantum
law continues to grow as \(\Theta(T)\).  For \(T\gtrsim d^2\), the two
models attain their respective geometric orders \(\Theta(d^2)\) and
\(\Theta(d)\).

The boundary \(\lambda_2=0\), \(0\le\lambda_1<1\), forms a distinct
false-acceptance regime. Every row then has a nonzero component accepted by its
designated test and annihilated by all other tests, which
yields \(N\le d\); a binary tree implemented by four permutations
matches this dimension bound together with the word-count bound.  Hence
all four models satisfy
\(
 \ln N=\Theta(\min\{\ln d,T\})
\)
throughout this face, as stated in Theorem~\ref{thm:zero}.  Every fixed
\(\lambda_2>0\) with \(\lambda_1+\lambda_2<1\) opens the positive-tolerance
packing regime described by Theorem~\ref{thm:main}.

\subsection{Explicit constructions and finite precision}

The four quantum transition maps are explicit: signed permutations are combined
with the rational rank-one update \(2uu^\top\). The four classical
blockwise permutation transitions are explicit block permutations. The probabilistic method
selects exponential word families from the corresponding bounded sets of words of length at most $T$
and thereby yields codebooks attaining the rates in
Theorem~\ref{thm:main}.

The unitary construction also has a finite-precision stability margin.
Suppose
\(\{\widetilde U_a:a\in\mathcal A_d\}\subset U(d)\) is an
inverse-closed unitary implementation satisfying
\(\max_a\|U_a-\widetilde U_a\|_{\mathrm{op}}\le\varepsilon\) and
\(\widetilde U_{a^{-1}}=\widetilde U_a^\dagger\) for every letter, with
the same convention for the self-inverse letters.  Telescoping a word
\(w\) of length at most \(L\) gives
\begin{equation}
 \|U_w-\widetilde U_w\|_{\mathrm{op}}\le L\varepsilon.
 \label{eq:word-robustness}
\end{equation}
Consequently, the normalized overlap of two conjugated rank-\(r\)
projections changes by at most \(4L\varepsilon\).  Running the packing
construction with internal target \(\lambda_2/2\) and taking
\(\varepsilon\le\lambda_2/(8L)\) preserves the off-diagonal bound, while the
exact inverse suffix preserves diagonal acceptance one.  Since
\(L=O_{\lambda_2}(d^2)\), an operator-norm tolerance of order \(d^{-2}\)
suffices for fixed error parameters; equivalently, the number of precision bits per scalar parameter needed
to specify the required tolerance grows as \(O(\log d)\).  Thus the complete
identification table is stable under inverse-closed unitary perturbations at
polynomially small operator-norm scale.

\subsection{Finite-precision and noise stability}

The same telescoping estimate gives a scaling analysis for
noisy implementations. The finite-precision statement concerns
perturbations of the intended
generators and preserves the inverse relation by assumption. An
implementation may apply a noisy channel after each transition, so that
the implemented word follows an approximate inverse relation.  If the
per-step implementation error is bounded in diamond norm by $\eta$, a
standard telescoping estimate gives an accumulated channel error of order
\(L\eta\) for a word of length \(L\). Preservation of the stated acceptance thresholds requires the accumulated error to
remain below a constant fraction of the available diagonal and off-diagonal
slacks; in particular, it is bounded by a constant multiple of
\(\min\{\lambda_1,\lambda_2\}\). For the present quantum construction,
\(L=O_{\lambda_2}(d^2)\), so a per-step noise scale of order
\(d^{-2}\) is a sufficient polynomial benchmark. A full theorem for
depolarizing, dephasing, and correlated non-Markovian noise would use a noisy
reachable-state version of the Grassmannian Hankel construction and is a natural
direction for future work.
\subsection{Extensions}

The framework suggests three groups of extensions. First, a circuit
realization problem fixes a dimension-independent local input alphabet,
charges circuit depth for synthesizing each transition, and incorporates
general noisy CPTP implementations with approximate inverse relations. This
replaces the maximum allowed word length by local-circuit complexity.

Second, explicit and efficiently indexable word families would
derandomize both achievability arguments.  The same combinatorial direction
includes determining the minimum alphabet cardinality that preserves the
quantum and classical scaling laws.

Third, fixing a supplied transition tuple, or drawing one from a specified
ensemble, defines individual and typical submatrix sizes for fixed or randomly drawn transition tuples. The
unary stochastic model is a natural first test case for this fixed-dynamics
program; resolving its scaling as a function of maximum word length requires sharp control of
transient behavior under a fixed separation margin.

Finally, we prove that four shared transitions suffice. Determining whether
two or three letters attain the same scaling as a function of maximum word length is an open
minimum-alphabet-size problem.  Lower bounds for
this question may have to combine growth restrictions for free semigroups
with generation constraints for permutation or unitary groups.  This
provides a concrete direction for reducing the input alphabet while
retaining the quantum--classical separation.

These extensions replace optimized growth of reachable sets under bounded word length by circuit, ensemble, or
noisy reachable-state geometry and lead to further scaling laws for
acceptance tables generated by shared transitions.

\section{Conclusion}
\label{sec:conclusion}

We developed an acceptance-table formulation for automata with a constant-size
input alphabet. Prefix words generate reachable states from a single initial
state, suffix words generate terminal effects or linear functionals, and the
resulting table is a finite prefix--suffix submatrix of the acceptance Hankel
matrix.

For every fixed $\lambda_2>0$ with $\lambda_1+\lambda_2<1$, we proved that
four shared transition maps suffice to achieve optimal state-space packing
scalings in the maximum word length. The unitary and general CPTP classes
satisfy $\ln N=\Theta_{\lambda_1,\lambda_2}(\min\{d^2,T\})$, whereas the
permutation and stochastic classes satisfy
$\ln N=\Theta_{\lambda_1,\lambda_2}(\min\{d,T\})$.
Explicit transition realizations were constructed: the four quantum transitions
support a Grassmannian packing via words of bounded length; the probabilistic
method yields an exponentially large word family for the encoded columns, while
inverse words supply the corresponding rows. Consequently, the quantum response
table is a transposed acceptance Hankel submatrix. On the classical side, four
explicit permutations generate the table through blockwise permutation encoding,
with the same probabilistic method supplying the codeword family. Exact inverse
words yield unit-diagonal acceptance, while subspace overlap and support
intersection govern the off-diagonal entries. The CPTP and stochastic converses
hold uniformly over arbitrary transition tuples, combining state-space packing
with word counting to establish that the quadratic and linear
maximum-word-length scales are optimal up to constants.

A qualitatively different regime emerges when $\lambda_2=0$ and
$0\leq\lambda_1<1$. Here each message must be accepted by exactly one test and
annihilated by all others, forcing $N\le d$ in both the quantum and classical
settings. A four-letter prefix-tree construction simultaneously matches this
dimension bound and the word-count bound, yielding
$\ln N=\Theta(\min\{\ln d,T\})$
across all four model classes. This reveals a structural dichotomy: metric
packing governs every fixed positive off-diagonal tolerance, whereas exclusive
disjoint accepting supports govern the zero false-acceptance case.

The transition families and their synthesis words are explicit; the exponential
codebook families are supplied by the probabilistic method and remain stable
under inverse-closed finite-precision perturbations at polynomially small
operator-norm scale. While four shared transitions suffice for the full
quantum--classical identification separation established here, whether two or
three maps are enough remains an open minimum-alphabet problem, alongside other
directions such as noisy-channel implementations, derandomized word families,
and table generation for a fixed supplied transition tuple.

\appendix

\clearpage

\section{Concentration on the Boolean Cube}
\label{app:boolean-concentration}

\begin{proof}[Proof of Lemma~\ref{lem:booleanLSI}]
The tensorized Bernoulli Gross log-Sobolev inequality has the following
normalization; see also the finite-chain treatment in
\cite{gross1975logarithmic,diaconis1996logarithmic}.
On one uniform bit, for values $a,b\in\mathbb R$,
\begin{equation}
 \Ent(g^2)
 \le\frac12(a-b)^2,
 \qquad g(+1)=a,\quad g(-1)=b,                    \label{eq:onebitLSI}
\end{equation}
where $\Ent(Y)=\mathbb E[Y\ln Y]-\mathbb EY\ln\mathbb EY$.
Replacing $a,b$ by their absolute values can only decrease the
right-hand side.  The case $a=b=0$ is immediate.  Otherwise, by
homogeneity normalize $(a^2+b^2)/2=1$ and, after interchanging $a$ and
$b$ if necessary, write
$a=\sqrt{1+x}$, $b=\sqrt{1-x}$, $0\le x\le1$.  The two sides
of \eqref{eq:onebitLSI} become
\begin{align}
 L(x)&=\tfrac12[(1+x)\ln(1+x)+(1-x)\ln(1-x)],\notag\\
 R(x)&=1-\sqrt{1-x^2}.
\end{align}
They agree at zero, while
\begin{equation}
 R'(x)-L'(x)=\frac{x}{\sqrt{1-x^2}}-\operatorname{arctanh}x\ge0,
                                                        \label{eq:LSIderivative}
\end{equation}
because the derivative of the last difference is
$(1-\sqrt{1-x^2})/(1-x^2)^{3/2}\ge0$.  This proves
\eqref{eq:onebitLSI}.

For a nonnegative function $Y$ of independent coordinates, entropy
tensorization and the log-sum inequality give
\begin{equation}
 \Ent(Y)\le\sum_{a=1}^N
 \mathbb E\,\Ent_a(Y),                             \label{eq:enttensor}
\end{equation}
where $\Ent_a$ takes entropy only over bit $a$, with all other bits
fixed.  To verify \eqref{eq:enttensor} inductively, write
\(X'=(X_1,\ldots,X_{N-1})\) and split off the last bit.  The exact chain
rule is
\(
 \Ent(Y)=\mathbb E_{X_N}\Ent_{X'}(Y)
 +\Ent_{X_N}(\mathbb E_{X'}Y).
\)
The induction hypothesis bounds the first term by the first \(N-1\)
terms on the right-hand side of \eqref{eq:enttensor}.  The map
\(h\mapsto\Ent_{X_N}(h)\) is convex on nonnegative two-vectors: this is
the two-term log-sum inequality.  Hence
\(
 \Ent_{X_N}(\mathbb E_{X'}Y)
 \le \mathbb E_{X'}\Ent_{X_N}(Y),
\)
which supplies the last term and proves \eqref{eq:enttensor}.
Applying \eqref{eq:onebitLSI} conditionally to $g$ and then
\eqref{eq:enttensor} yields
\begin{equation}
 \Ent(g^2)\le\frac12\mathbb E
 \sum_{a=1}^N(g(X)-g(X^{(a)}))^2.                 \label{eq:cubeLSI}
\end{equation}

Set $g=e^{\theta f/2}$.  The elementary inequality
\begin{equation}
 (e^{x/2}-e^{y/2})^2
 \le\frac18(x-y)^2(e^x+e^y)                      \label{eq:expineq}
\end{equation}
follows by putting \(z=(x-y)/2\), using
\(\sinh^2(z/2)\le(z^2/4)\cosh^2(z/2)\), and
\(\cosh^2(z/2)\le\cosh z\).  Inequality \eqref{eq:expineq} and
symmetry under $X\leftrightarrow X^{(a)}$ turn
\eqref{eq:cubeLSI} and \eqref{eq:gradV} into
\begin{equation}
 \Ent(e^{\theta f})
 \le\frac{\theta^2V}{8}\mathbb E e^{\theta f}.   \label{eq:mgfentropy}
\end{equation}
For
$\psi(\theta)=\ln\mathbb E e^{\theta(f-\mathbb Ef)}$, the identity
$(\psi(\theta)/\theta)'=
 \Ent(e^{\theta f})/(\theta^2\mathbb E e^{\theta f})$
and \eqref{eq:mgfentropy} give, by integration from \(0\) to \(\theta\),
\(
 \frac{\psi(\theta)}{\theta}
 \le\int_0^\theta\frac{V}{8}\,ds=\frac{\theta V}{8},
 \qquad
 \psi(\theta)\le\frac{\theta^2V}{8}.
\)
Chernoff's inequality and the choice
$\theta=4t/V$ prove \eqref{eq:booleanconc}.
\end{proof}

\section{Squared-Gradient Estimate}
\label{app:squared-gradient}

\begin{proof}[Proof of Lemma~\ref{lem:reflection-concentration}]
Fix every sign vector in \eqref{eq:Fm} except $s_k$.  Absorbing all
preceding and following orthogonal factors into two rank-$r$
projections $R$ and $A$, the dependence on
$v=d^{-1/2}s_k$ has the form
\begin{equation}
 f(v)=\frac1r\Tr(AH_vRH_v),
 \qquad H_v=I-2vv^\top.                    \label{eq:localf}
\end{equation}
Expansion gives
\begin{equation}
 f(v)=\text{constant}+\frac4r
 \big[(v^\top Av)(v^\top Rv)
      -v^\top ARv\big].                    \label{eq:fexpand}
\end{equation}
For a matrix $M$, put $M_{\rm s}=(M+M^\top)/2$ and
$q_M(v)=v^\top Mv$.  Then $q_M=q_{M_{\rm s}}$.  In the applications
below, $\lVert M_{\rm s}\rVert_{\mathrm{op}}\le1$, so it suffices to take
$M=M^\top$ with $\lVert M\rVert_{\mathrm{op}}\le1$.  Flipping
coordinate $a$ then gives, with $v^{(a)}$ denoting the resulting vector,
\begin{equation}
 q_M(v^{(a)})-q_M(v)
 =-4v_a(Mv)_a+4v_a^2M_{aa}.                       \label{eq:qflip}
\end{equation}
Since $v_a^2=1/d$, $\lVert Mv\rVert_2\le1$, and
$|M_{aa}|\le1$,
\begin{equation}
 \sum_{a=1}^d(q_M(v)-q_M(v^{(a)}))^2
 \le\frac{64}{d}.                                \label{eq:qgrad}
\end{equation}
All three quadratic forms in \eqref{eq:fexpand} have absolute value at
most one; for $AR$, use its symmetric part
$(AR+RA)/2$, whose operator norm is at most one.  If
\(\Delta_aq=q(v)-q(v^{(a)})\), the finite-difference product rule and
\((x+y)^2\le2x^2+2y^2\) give explicitly, with
\(h=q_Aq_R-q_{AR}\),
\begin{align}
 \sum_a[\Delta_a(q_Aq_R)]^2
 &\le2\sum_a(\Delta_aq_A)^2\notag\\
 &\quad+2\sum_a(\Delta_aq_R)^2
 \le\frac{256}{d},                                      \label{eq:productgrad}\\
 \sum_a(\Delta_a h)^2
 &\le2\frac{256}{d}+2\frac{64}{d}
 =\frac{640}{d}.                                        \label{eq:combinedgrad}
\end{align}
Multiplication by the prefactor $(4/r)^2$ in
\eqref{eq:fexpand} therefore gives
\begin{equation}
 \sum_{a=1}^d(f(v)-f(v^{(a)}))^2
 \le\frac{10240}{r^2d}.                           \label{eq:localgradient}
\end{equation}

There are $m$ independent sign vectors and the bound is uniform after
all other signs are fixed.  Consequently, as a function of all $md$
sign bits, let $F_m^{(k,a)}$ denote $F_m$ after flipping coordinate $a$
of the $k$th sign vector.  Then
\begin{equation}
 \sup\sum_{k=1}^m\sum_{a=1}^d
 (F_m-F_m^{(k,a)})^2
 \le\frac{10240m}{r^2d}.                          \label{eq:globalgradient}
\end{equation}
Lemma~\ref{lem:booleanLSI} now implies \eqref{eq:Ftail}.
\end{proof}

\section*{Acknowledgment}
The author would like to thank Fei Yan and Pan Zhang for their mentorship and guidance. Deep gratitude is also extended to the author's wife, Ling Wang, for her unwavering support and encouragement.

\section*{Declaration of competing interest}
The author declares that there are no known competing financial interests or
personal relationships that could have appeared to influence the work reported
in this paper.


\section*{Declaration of generative AI and AI-assisted technologies in the manuscript preparation process}
During the preparation of this manuscript, the author used OpenAI Codex to
assist with language editing, structural polishing, and mathematical-
consistency checks. The author reviewed and edited all generated output and
takes full responsibility for the content of the manuscript.

\bibliographystyle{elsarticle-num}
\bibliography{refs}

\end{document}